\documentclass{article}

\usepackage{times}
\usepackage{graphicx} 
\usepackage{subfigure} 

\usepackage{natbib}

\usepackage{algorithm}
\usepackage{algorithmic}

\usepackage{mathtools, amsmath, amsfonts, amssymb, amsthm}

\DeclareMathOperator*{\argmin}{arg\,min}
\newcommand{\norm}[1]{\left\lVert#1\right\rVert}

\newtheorem{theorem}{Theorem}[section]

\newtheorem{assumption}[theorem]{Assumption}
\theoremstyle{remark}

\numberwithin{equation}{section}

\theoremstyle{remark}
\newtheorem*{remark}{Remark}

\usepackage{hyperref}

\usepackage[accepted]{icml2016}

\icmltitlerunning{Balancing Method for High Dimensional Causal Inference}

\begin{document} 

\twocolumn[
\icmltitle{Balancing Method for High Dimensional Causal Inference}

\icmlauthor{Thai T. Pham}{thaipham@stanford.edu}
\icmladdress{Stanford Graduate School of Business,
            655 Knight Way, Stanford, CA 94305 USA}

\icmlkeywords{causal inference, individual treatment effect, high-dimension, factual- counterfactual balancing}

\vskip 0.3in
]

\begin{abstract} 

Causal inference has received great attention across different fields from economics, statistics, education, medicine, to machine learning. Within this area, inferring causal effects at individual level in observational studies has become an important task, especially in high dimensional settings. In this paper, we propose a framework for estimating Individualized Treatment Effects in high-dimensional non-experimental data. We provide both theoretical and empirical justifications, the latter by comparing our framework with current best-performing methods. Our proposed framework rivals the state-of-the-art methods in most settings and even outperforms them while being much simpler and easier to implement.

\end{abstract} 

\section{Introduction}
\label{sec:introduction}

Apparently in order to make decisions, most of the time we fundamentally need to answer causal questions. To decide whether to take a particular medication or not, patients need to know the effect of the medication on outcome (e.g., health status), i.e., they need to know what the outcome would be if they take and do not take the medication. Similarly, to decide whether to buy a car, a buyer needs to know the outcome (e.g., happiness over cost pain) when s/he buys relative to the outcome when s/he does not. As we may notice, all these decisions need to be made at individual level, for each patient and for each buyer. 

In these questions, the decision to take the medication or to buy the car is the intervention, which is often called treatment in social science studies. Scientists have been trying to answer this kind of challenging questions with real data. The problem becomes harder when we have only observational data, and recently when the data is high-dimensional. Working with observational data is difficult because in this type of data, we have no control over the treatment decision making process. Thus, a lot of biases can exist; for example, the patients who previously took the medication could fundamentally be different from those who did not. These biases would make causal inference harder to deal with. 

In this paper, we propose a simple but effective framework for estimating Individualized Treatment Effects (ITEs) in high-dimensional observational data. Our proposed framework is inspired by the ideas from domain adaptation, particularly the covariate shift problem. The first step is to find the balancing weights mapping the groups of individuals with the same treatment decision to a common targeted population. We will explain in detail in Section \ref{sec:cBal} what this targeted population is and how to find the weights. This step comes naturally after the observation we made earlier that the individuals with different treatment decisions are usually fundamentally different; by using balancing weights, we can make the individuals look more like one another. Then with these weights, we can estimate the ITEs via some forms of weighted regressions. The detailed procedures are provided in Algorithms \ref{alg:counterfactualBal} and \ref{alg:factual-counterlBal}.  

As \cite{johansson2016learning} point out, counterfactual questions appear in machine learning under the ``bandit feedback'' setting in, for example, off-policy evaluation in reinforcement learning \cite{beygelzimer2010contextual, sutton1998reinforcement}. However in such settings, the researcher typically knows the method underlying the treatment decision making process (such as the policy function in reinforcement learning). This is much easier than our currently considered problem in which we do not have such informational merits since our data is fully observational and only real values with no functional form are known.

With the rise of big data, more and more high-dimensional non-experimental datasets become available to researchers; at the same time, researchers are urged to utilize these datasets in meaningful ways to solve real problems. In joining this task force, we propose easy-to-implement, yet powerful methods for estimating ITEs using such data through which to help policy makers, companies, and individuals, etc. make correct decisions. 

Our contributions are as follows. First, we provide a powerful, yet simple framework to estimate ITEs with observational data of high dimensions in Sections \ref{sec:cBal} and \ref{sec:fcBal}. Second, we provide theoretical error bounds on the estimated ITEs in Section \ref{sec:theory}. Third, we demonstrate good empirical results on four simulated datasets in Section \ref{sec:experiment}.

\section{Problem Setup}
\label{sec:prob_setup}

Assume that there are $n$ observed individuals. Each of them has a set of information which could be age, salary, race, religion, purchasing history, etc.; we denote by $X$ this set of features. Each individual also made a decision on treatment $W$; in this paper, we consider binary treatment. Treatments could be any decision ranging from taking a medication, buying a house, to eating at a specific restaurant. We know only the realized values of $W$ without knowing how the decision was made. Finally, each individual is associated with an observed outcome $Y$ (i.e., return on the decision made). Now a new individual with only observed features $X$ comes in and needs to make a treatment decision. The question is how we can use the data on $n$ observed individuals to help the new individual make the correct decision. We frame this as an ITE estimation problem below.  

We have observed data $(Y_i, W_i, X_i)_{i = 1}^n$, where $n$ the number of observations. For each $i \in \{1, ..., n\}$, $X_i \in \mathbb{R}^d$ is the set of features, $W_i \in \{0, 1\}$ is the treatment indicator, and $Y_i \in \mathbb{R}$ is the outcome of interest. As we consider high-dimensional setting, $d$ is really large relative to $n$. 

The individual $i$ with $W_i = 1$ is usually referred to as \textit{treated} while the one with $W_i = 0$ is referred to as \textit{control}. We write $Y_i = Y_i(W_i)$ to denote the \textit{potential factual outcome} and $Y_i(1 - W_i)$ to denote the \textit{potential counterfactual outcome}, i.e. the outcome for $i$ had that individual followed treatment decision $(1 - W_i)$ instead of $W_i$. The potential outcome notation is used following the Rubin Causal Model (see \cite{imbens2015causal} for an overview). We also denote $\displaystyle n_t = \sum_{i = 1}^n W_i$ and $n_c = n - n_t$. 

We are interested in estimating the ITE of $W$ on $Y$ given $X$:
\begin{equation*}
	\tau(x) = \mathbb{E}[Y(1) - Y(0) | X = x].
\end{equation*}
This quantity is of interest in many fields because it allows people to make the best decision when they face with two options. For example, a doctor needs to decide whether to give a considered medication to a particular patient; a school needs to decide whether to assign a particular student to an honor program or the standard one; etc. 

Note that for each individual $i$, exactly one of $Y_i(1)$ and $Y_i(0)$ is observed; if one chose $W_i = 1$, then one could not choose $W_i = 0$ and vice versa. This is the fundamental problem of causal inference \cite{holland1986statistics}. As mentioned earlier, this is known as ``bandit feedback'' in machine learning. This implies that the outcome of interest $Y_i(1) - Y_i(0)$ is never known for any individual $i$; this makes causal inference hugely different from supervised learning.   

One important remark here is the Stable Unit Treatment Value Assumption (SUTVA). When we write $Y(1)$ and $Y(0)$ we implicitly make SUTVA, which assumes that one's outcome is not affected by others' treatment decisions (no interference) and the value of the treatment is the same across treated individuals (no variation in treatment value).\footnote{Note that it is not to say there is no variation in treatment effect; in fact, the opposite is true.} The first part of this assumption could well be violated, for example, in a network setting when peer or spillover effects exist. The second part of the assumption could be violated too, for example, in a medical treatment in which the same medicine contains some stronger ingredients for a subset of patients and weaker components for others. In this work, we assume away these two issues and focus on applications when such issues are not present. 

We also need to make the unconfoundedness assumption as usually done in causal inference literature. 

\begin{assumption}
\label{ass:unconfoundedness} (Unconfoundedness)

Conditional on observed features, the potential outcomes are independent of treatment, i.e. 
\begin{equation*}
	Y(0), Y(1) \perp W | X.
\end{equation*}
\end{assumption}
This assumption means no other feature could affect $W$ and $Y$. This assumption allows one to assign the cause of the effect $Y(1) - Y(0)$ to only treatment $W$.  

Now, we make assumptions regarding the distributions of the data. We assume that we observe $\{X_i | W_i = 1\}$ drawn i.i.d. from $q_t(x)$ associated with $\{Y_i | W_i = 1\} = \{Y_i(1) | W_i = 1\}$ drawn i.i.d. from $q_t(y | x)$; $\{X_i | W_i = 0\}$ drawn i.i.d. from $q_c(x)$ associated with $\{Y_i | W_i = 0\} = \{Y_i(0) | W_i = 0\}$ drawn i.i.d. from $q_c(y | x)$. Assume that we also observe $\{X_i | i \in I_T\}$ over a targeted population $I_T$ on which we want to estimate the treatment effect. As we will see later, the averages of the features along each dimension on $I_T$ are meant to be used as targets for the weighted treated and control groups to match. We assume that these $\{X_i | i \in I_T\}$ are drawn i.i.d. from $p(x)$, and are observed. However, we observe neither $\{Y_i(1) | i \in I_T\}$ which are drawn i.i.d. from $p_t(y|x)$ nor $\{Y_i(0) | i \in I_T\}$ which are drawn i.i.d. from $p_c(y|x)$. 

In practice, we use all observed data $I_T = \{i | 1 \leq i \leq n\}$ as the targeted population. In this case, part of $\{Y_i(1) | i \in I_T\}$ and $\{Y_i(0) | i \in I_T\}$ is observed. Moreover, part of the feature data $\{X_i | i \in I_T\}$ is assumed to have distribution $q_t(x)$ while the other part has distribution $q_c(x)$; this seems to contradict our assumption above. However, We still use all these available data as the targeted population because we will use only the averages of the features on this data set for matching purpose; using all data would allow us to have better match as we look at the whole population.

In deriving the error bounds in the theory section, however, we still assume all feature data on the targeted population comes from the same distribution $p(x)$ and no (partial) outcome is observed. 

The last implicit assumption regards the distribution of outcomes in the targeted population.

\begin{assumption}
\label{ass:sameDistr} (Potential Outcome Distribution)

The conditional distributions of treated and control groups are the same as that of the targeted population, i.e. 
\begin{equation*}
	p_t(y | x) = q_t(y | x) \text{ and } p_c(y | x) = q_c(y | x).
\end{equation*}
\end{assumption}
If this assumption is not satisfied, then no algorithm could give a sound estimate for the ITEs.

\section{Counterfactual Balancing}
\label{sec:cBal}

\begin{algorithm}[tb]
  \caption{Counterfactual Balancing}
  \label{alg:counterfactualBal}
\begin{algorithmic}[1]
  \STATE {\bfseries Input:} $(Y_i, W_i, X_i)_{i = 1}^n$
  \STATE Determine balancing weights $\eta \in \mathbb{R}^{n_t}$, $\gamma \in \mathbb{R}^{n_c}$ as
\begin{equation*}
	\eta = \argmin_{\widetilde{\eta} \in \mathbb{R}^{n_t}} F_t(\widetilde{\eta}; \xi, r) \text{ s.t. } \sum_i \widetilde{\eta}_i = 1 \text{ and each } \widetilde{\eta}_i \geq 0
\end{equation*}
\vskip-0.75cm
\begin{equation*}
	\gamma = \argmin_{\widetilde{\gamma} \in \mathbb{R}^{n_c}} F_c(\widetilde{\gamma}; \xi, r) \text{ s.t. } \sum_i \widetilde{\gamma}_i = 1 \text{ and each } \widetilde{\gamma}_i \geq 0
\end{equation*}
  \STATE Fit 
  \vskip-0.5cm
  \begin{equation*}
	\widehat{\beta}_{p_t, \lambda_{p_t}} = \argmin_{\beta} G_t(\beta; \lambda_{p_t}, \alpha)
  \end{equation*}
  \vskip-0.25cm
  \begin{equation*}
	\widehat{\beta}_{p_c, \lambda_{p_c}} = \argmin_{\beta} G_c(\beta; \lambda_{p_c}, \alpha)
\end{equation*}
  \STATE {\bfseries Output:} $\widehat{\beta}_{p_t, \lambda_{p_t}}, \widehat{\beta}_{p_c, \lambda_{p_c}}$. 
  \vskip+0.1cm
  \STATE {\bfseries Inference:} For a new individual with only observed features $x^{new}$, the estimated ITE is
\vskip-0.25cm
\begin{equation*}
	\widehat{\tau}(x^{new}) = x^{new} (\widehat{\beta}_{p_t, \lambda_{p_t}} - \widehat{\beta}_{p_c, \lambda_{p_c}}).
\end{equation*}
\end{algorithmic}
\end{algorithm}

Because the observed data and the targeted population over which we want to estimate the ITEs are distributed differently, it is natural to balance these groups with balancing weights to move one group close to the other in terms of distribution. 

In this section, we present a simple procedure for matching two groups of individuals in terms of feature distribution and using that to estimate the ITEs. Balancing weights has a long literature in machine learning, specifically in covariate shift problems \cite{reddi2015doubly} as well as in causal inference area \cite{zubizarreta2015stable, athey2016approximate} among others; the latter focuses mostly on average treatment effect estimation though.  

We proceed by denoting by $\overline{X} \in \mathbb{R}^d$ the column means of $\{X_i | i \in I_T\}$, by $\mathbf{X}_t$ the feature matrix of the treated individuals and $\mathbf{X}_c$ the feature matrix of the control ones. The counterfactual balancing algorithm for ITEs is provided in Algorithm \ref{alg:counterfactualBal}. Here, 
\begin{equation*}
	F_t(\widetilde{\eta}; \xi, r) = (1 - \xi) \norm{\widetilde{\eta}}_2^2 + \xi \norm{\overline{X} - \mathbf{X}_t^T \widetilde{\eta}}_r^2
\end{equation*}
\vskip-0.5cm
\begin{equation*}
	F_c(\widetilde{\gamma}; \xi, r) = (1 - \xi) \norm{\widetilde{\gamma}}_2^2 + \xi \norm{\overline{X} - \mathbf{X}_t^T \widetilde{\gamma}}_r^2,
\end{equation*}
where the hyperparameter $\xi$ is usually chosen to be $0.5$. We consider both $r = 2$ and $r = \infty$. Also,
\begin{multline*}
	G_t(\beta; \lambda_{p_t}, \alpha) = \frac{1}{n_t} \sum_{i: W_i = 1} \eta_i (Y_i(1) - X_i \beta)^2 + \\
    \lambda_{p_t} \big( (1 - \alpha) \norm{\beta}_2^2 + 2 \alpha \norm{\beta}_1 \big)
\end{multline*}
\vskip-0.75cm
\begin{multline*}
	G_c(\beta; \lambda_{p_c}, \alpha) = \frac{1}{n_c} \sum_{i: W_i = 0} \gamma_i (Y_i(0) - X_i \beta)^2 + \\
    \lambda_{p_c} \big( (1 - \alpha) \norm{\beta}_2^2 + 2 \alpha \norm{\beta}_1 \big).
\end{multline*}
In practice, we usually choose $\alpha = 0.9$ and select $\lambda_{p_t}, \lambda_{p_c}$ by cross-validation.

\section{Factual-Counterfactual Balancing}
\label{sec:fcBal}

\begin{algorithm}[tb]
  \caption{Factual-Counterfactual Balancing}
  \label{alg:factual-counterlBal}
\begin{algorithmic}[1]
  \STATE {\bfseries Input:} $(Y_i, W_i, X_i)_{i = 1}^n$
  \STATE Determine balancing weights $\eta \in \mathbb{R}^{n_t}$, $\gamma \in \mathbb{R}^{n_c}$ as
\begin{equation*}
	\eta = \argmin_{\widetilde{\eta} \in \mathbb{R}^{n_t}} F_t(\widetilde{\eta}; \xi, r) \text{ s.t. } \sum_i \widetilde{\eta}_i = 1 \; \& \text{ each } \widetilde{\eta}_i \geq 0
\end{equation*}
\vskip-0.75cm
\begin{equation*}
	\gamma = \argmin_{\widetilde{\gamma} \in \mathbb{R}^{n_c}} F_c(\widetilde{\gamma}; \xi, r) \text{ s.t. } \sum_i \widetilde{\gamma}_i = 1 \; \& \text{ each } \widetilde{\gamma}_i \geq 0
\end{equation*}
  \STATE Fit 
  \vskip-0.5cm
  \begin{equation*}
	\widehat{\beta}_{q_t, \lambda_{q_t}} = \argmin_{\beta} H_t(\beta; \lambda_{q_t}, \alpha)
\end{equation*}
\vskip-0.25cm
\begin{equation*}
	\widehat{\beta}_{q_c, \lambda_{q_c}} = \argmin_{\beta} H_c(\beta; \lambda_{q_c}, \alpha)
\end{equation*}
  \STATE Determine the final estimates: 
  \vskip-0.5cm
\begin{multline*}
	\widehat{\beta}_t^{final} = \argmin_\beta \frac{1}{n_t} \sum_{i: W_i = 1} \eta_i (Y_i(1) - X_i \beta)^2 + \\ 
    \lambda_t^\prime || \beta - \widehat{\beta}_{q_t, \lambda_{q_t}} ||_2^2
\end{multline*}
\vskip-0.75cm
\begin{multline*}
	\widehat{\beta}_c^{final} = \argmin_{\beta} \frac{1}{n_c} \sum_{i: W_i = 0} \gamma_i (Y_i(0) - X_i \beta)^2 + \\ 
    \lambda_c^\prime || \beta - \widehat{\beta}_{q_c, \lambda_{q_c}} ||_2^2
\end{multline*}
  \STATE {\bfseries Output:} $\widehat{\beta}_t^{final}, \widehat{\beta}_c^{final}$. 
  \vskip+0.1cm
  \STATE {\bfseries Inference:} For a new individual with only observed features $x^{new}$, the estimated ITE is
  \vskip-0.25cm
\begin{equation*}
	\widehat{\tau}(x^{new}) = x^{new} (\widehat{\beta}_t^{final} - \widehat{\beta}_c^{final}).
\end{equation*}
\end{algorithmic}
\end{algorithm}

In this section, we present the second algorithm (Algorithm \ref{alg:factual-counterlBal}) for ITE estimation. Here,
\begin{multline*}
	H_t(\beta; \lambda_{q_t}, \alpha) = \frac{1}{n_t} \sum_{i: W_i = 1} (Y_i(1) - X_i \beta)^2 + \\
    \lambda_{q_t} \big( (1 - \alpha) \norm{\beta}_2^2 + 2 \alpha \norm{\beta}_1 \big)
\end{multline*}
\vskip-0.75cm
\begin{multline*}
	H_c(\beta; \lambda_{q_c}, \alpha) = \frac{1}{n_c} \sum_{i: W_i = 0} (Y_i(0) - X_i \beta)^2 + \\
    \lambda_{q_c} \big( (1 - \alpha) \norm{\beta}_2^2 + 2 \alpha \norm{\beta}_1 \big)
\end{multline*}
The objective function in Step $4$ consists of two terms: a weighted estimation function and a constraint. The first term is meant to make the counterfactual prediction accurate while the constraint is meant to keep the final estimators not too far from the estimators obtained in Step $3$ through which to make the factual prediction accurate. If we include only the first term in Step 4, then we should have good counterfactual prediction but at the same time, we lose signals. Moreover, the estimator in the usual weighted regression is known to be unstable; adding the constraint would make it closer to the estimator in the unweighted regression, which helps increasing stability.

The idea of trading off between factual and counterfactual outcome predictions appears in the literature. For example, \cite{johansson2016learning} use an objective function comprising of factual errors, counterfactual errors measured by nearest neighbor values from the other group, and the discrepancy distance between the distributions of two groups.

\section{Theoretical Results}
\label{sec:theory}

Our error bound results are inspired by \cite{reddi2015doubly}. First of all, we obtain an error bound between the treated group and the targeted population. The error bound between the control group and the targeted population can be established in a similar way. Combining these two error bounds, we can obtain an error bound on the ITE estimation. 

The bound between a group (treated or control) and the targeted population is on the relative error for a model fit on the outcomes of that group and evaluated on the targeted population, as compared with that if it was fit on the targeted population. In this section, we derive error bounds with $r = 2$ and $\alpha = 0$ in the proposed algorithms. We also define
\begin{equation*}
	\beta_{p_t, \lambda_{p_t}}^* = \argmin_{\beta} R_{p_t} [\beta] + \lambda_{p_t} ||\beta||_2^2
\end{equation*}
\vskip-0.5cm
\begin{equation*}
	\beta_{q_t, \lambda_{q_t}}^*= \argmin_\beta R_{q_t}[\beta] + \lambda_{q_t} ||\beta||_2^2
\end{equation*}
where 
\begin{equation*}
	R_{p_t} [\beta] = \mathbb{E}_{(Y(1), X) \sim p_t} ||Y(1) - X \beta ||_2^2
\end{equation*}
\begin{equation*}
	R_{q_t} [\beta] = \mathbb{E}_{(Y(1), X) \sim q_t} ||Y(1) - X \beta ||_2^2.
\end{equation*}
and
\begin{equation*}
	\beta_{p_c, \lambda_{p_c}}^* = \argmin_{\beta} R_{p_c} [\beta] + \lambda_{p_c} ||\beta||_2^2
\end{equation*}
\vskip-0.5cm
\begin{equation*}
	\beta_{q_c, \lambda_{q_c}}^*= \argmin_\beta R_{q_c}[\beta] + \lambda_{q_c} ||\beta||_2^2
\end{equation*}
where
\begin{equation*}
	R_{p_c} [\beta] = \mathbb{E}_{(Y(0), X) \sim p_c} ||Y(0) - X \beta ||_2^2
\end{equation*}
\begin{equation*}
	R_{q_c} [\beta] = \mathbb{E}_{(Y(0), X) \sim q_c} ||Y(0) - X \beta ||_2^2.
\end{equation*}

To obtain error bounds, we need one more assumption regarding the balancing weights.

\begin{assumption}
\label{ass:weight_exist} (Balancing Weights)

The true balancing weights $\eta$ and $\gamma$ are well defined and bounded above by $C$ (i.e. $\eta_i, \gamma_j  \leq C$ for all $i, j$).\footnote{For example, $C = 1$ satisfies this assumption.} 
\end{assumption}

Next, we present the theoretical results on error bounds for the ITE estimator in Algorithm \ref{alg:counterfactualBal}. We start with bounds on the treated and control groups separately. 

\begin{theorem}
\label{thm:alg_1_error_bound_1}
Assume \ref{ass:unconfoundedness}, \ref{ass:sameDistr}, and \ref{ass:weight_exist}. For $w \in \{t, c\}$, the following holds with probability at least $(1 - \delta)$:
\begin{equation*}
	R_{p_w} [\widehat{\beta}_{p_w, \lambda_{p_w}}] \leq R_{p_w} [\beta_{p_w, \lambda_{p_w}}^*] + \Delta_S^w + \Delta_R^w,
\end{equation*}
where 
\begin{equation*}
	\Delta_S^w = \frac{L^2 \kappa^{\frac{3}{2}} \sigma^{\frac{1}{2}}(\mathbf{K}_w)}{n_w \lambda_{p_w}} \sqrt{\frac{(C n_w)^2 + 1}{n_w}} \left(1 + \sqrt{2 \log \left( \frac{4}{\delta} \right)}\right)
\end{equation*}
\begin{equation*}
	\Delta_R^w = \frac{4 (C n_w) L \nu_{p_w} \sqrt{tr(\mathbf{K}_w)}}{n_w} + 6 (Cn_w) L \sqrt{\frac{\log (4 / \delta)}{2n_w}}.
\end{equation*}
Here, we write $\mathbb{E}_p$ to mean the expectation over $X \sim p(x)$, $Y(1) \sim p_t(y | x)$, $Y(0) \sim p_c(y | x)$. Also, $\nu_{p_w} = \big|\big|\beta_{p_w, \lambda_{p_w}}^*\big|\big|_2^2$; $L$ is the upper bound of $||\mathbb{E}_{Y | X} || Y - X \beta ||_2^2 ||_\infty$; $\kappa$ is the upper bound of $\norm{X_i}_2$ for all $i$; $\mathbf{K}_w$ is the identity kernel matrix for the treated and control features; $\sigma(\mathbf{K}_w)$ is the condition number of matrix $\mathbf{K}_w$ (i.e. the ratio of the largest to smallest singular value in the singular value decomposition of matrix $\mathbf{K}_w$); and $tr(\mathbf{K}_w)$ is the trace of $\mathbf{K}_w$, which in this case is $\sum_{i: W_i = 1} \norm{X_i}_2^2$ when $w = t$ and $\sum_{i: W_i = 0} \norm{X_i}_2^2$ when $w = c$. 
\end{theorem}
\begin{proof}
This is a direct application of Theorem $5$ in \cite{reddi2015doubly} (full paper) with balancing weights $(n_t \eta)$ and $(n_c \gamma)$, hyperparameters $(n_t \lambda_{p_t})$ and $(n_c \lambda_{p_c})$, and bounds $(C n_t)$ and $(C n_c)$.
\end{proof}

To establish error bounds on ITE estimation, we need two other assumptions. The first one regards the sufficiency of observed features in explaining the dependence between the potential outcomes.

\begin{assumption}
\label{ass:indep_PO} (Potential Outcome Independence)

The potential outcomes are independent given observed features, i.e. 
\begin{equation*}
	Y(0) \perp Y(1) | X.
\end{equation*}
\end{assumption}

The second one regards the true functional form of $Y(1)$ and $Y(0)$. 
\begin{assumption}
\label{ass:true_form_PO}
(Linearity)
\begin{align*}
	\mathbb{E}_{p_t(y | x)} [Y(1) | X] &= X \beta_{p_t, \lambda_{p_t}}^*. \\
    	\mathbb{E}_{p_c(y | x)} [Y(0) | X] &= X \beta_{p_c, \lambda_{p_c}}^*.
\end{align*}
\end{assumption}
Also using this linearity assumption, \cite{athey2016approximate} note that this assumption is strong, but it may be plausible as the researcher could include transformations of the basic features in the design or use a pre-training step.

Now, we can give an error bound to the ITE estimator. 
\begin{theorem}
\label{thm:alg_1_error_bound}
Under Assumptions \ref{ass:unconfoundedness}, \ref{ass:sameDistr}, \ref{ass:weight_exist}, \ref{ass:indep_PO}, \ref{ass:true_form_PO}, the following holds with probability at least $(1 - \delta)$:
\begin{multline*}
	\mathbb{E}_p \norm{(Y(1) - Y(0)) - (X \widehat{\beta}_{p_t, \lambda_{p_t}} - X \widehat{\beta}_{p_c, \lambda_{p_c}})}_2^2 \leq \\
    2\mathbb{E}_p \norm{(Y(1) - Y(0)) - (X \beta_{p_t, \lambda_{p_t}}^* - X \beta_{p_c, \lambda_{p_c}}^*)}_2^2 \\
    + 2\left(\Delta_S^t + \Delta_R^t + \Delta_S^c + \Delta_R^c\right),
\end{multline*}
where notations are defined in Theorem \ref{thm:alg_1_error_bound_1}.
\end{theorem}
\begin{proof}
By the Cauchy-Schwarz inequality and Theorem \ref{thm:alg_1_error_bound_1}, we have
\begin{align*}
& \mathbb{E}_p \norm{(Y(1) - Y(0)) - (X \widehat{\beta}_{p_t, \lambda_{p_t}} - X \widehat{\beta}_{p_c, \lambda_{p_c}})}_2^2 \\
&\leq 2 \left( R_{p_t}[\widehat{\beta}_{p_t, \lambda_{p_t}}] + R_{p_c}[\widehat{\beta}_{p_c, \lambda_{p_c}}] \right) \\
&\leq 2 \left(R_{p_t}[\beta_{p_t, \lambda{p_t}}^*] + R_{p_c}[\beta_{p_c, \lambda_{p_c}}^*]\right) \\
& \hskip+3cm + 2\left(\Delta_S^t + \Delta_R^t + \Delta_S^c + \Delta_R^c\right).
\end{align*}
Now, Assumptions \ref{ass:indep_PO} and \ref{ass:true_form_PO} imply
\begin{align*}
& \mathbb{E}_p \norm{(Y(1) - Y(0)) - (X \beta_{p_t, \lambda_{p_t}}^* - X \beta_{p_c, \lambda_{p_c}}^*)}_2^2 \\
&= R_{p_t}[\beta_{p_t, \lambda{p_t}}^*] + R_{p_c}[\beta_{p_c, \lambda_{p_c}}^*] \\
& \hskip+0.5cm -2 \mathbb{E}_{p_t}[ Y(1) - X \beta_{p_t, \lambda_{p_t}}^*] \times \mathbb{E}_{p_c}[ Y(0) - X \beta_{p_c, \lambda_{p_c}}^*] \\
&= R_{p_t}[\beta_{p_t, \lambda{p_t}}^*] + R_{p_c}[\beta_{p_c, \lambda_{p_c}}^*].
\end{align*}
This ends the proof.
\end{proof}

\begin{remark}
\label{rm:error_bound_1}
We can bound the term in Theorem \ref{thm:alg_1_error_bound} in a different way by noting that 
\begin{align*}
	\mathbb{E}_{p_t(y | x)} [Y(1) - X \widehat{\beta}_{p_t, \lambda_{p_t}}] = X (\beta_{p_t, \lambda_{p_t}}^* - \widehat{\beta}_{p_t, \lambda_{p_t}}) \\
    	\mathbb{E}_{p_c(y | x)} [Y(0) - X \widehat{\beta}_{p_c, \lambda_{p_c}}] = X (\beta_{p_c, \lambda_{p_c}}^* - \widehat{\beta}_{p_c, \lambda_{p_c}}).
\end{align*}
Using the similar technique to that above, we can obtain
\begin{multline*}
	\mathbb{E}_p \norm{(Y(1) - Y(0)) - (X \widehat{\beta}_{p_t, \lambda_{p_t}} - X \widehat{\beta}_{p_c, \lambda_{p_c}})}_2^2 \leq \\
    \mathbb{E}_p \norm{(Y(1) - Y(0)) - (X \beta_{p_t, \lambda_{p_t}}^* - X \beta_{p_c, \lambda_{p_c}}^*)}_2^2 \\
    + \left(\Delta_S^t + \Delta_R^t + \Delta_S^c + \Delta_R^c\right) \\
    - 2 \mathbb{E}_{p(x)}[X (\beta_{p_t, \lambda_{p_t}}^* - \widehat{\beta}_{p_t, \lambda_{p_t}})  X (\beta_{p_c, \lambda_{p_c}}^* - \widehat{\beta}_{p_c, \lambda_{p_c}})].
\end{multline*}
This bound is less favorable because the terms $\widehat{\beta}_{p_t, \lambda_{p_t}}$ and $\widehat{\beta}_{p_c, \lambda_{p_c}}$ are still present. However, if we can bound these terms by $\beta_{p_t, \lambda_{p_t}}^*$ and $\beta_{p_c, \lambda_{p_c}}^*$ respectively then this error bound will actually be more favorable.
\end{remark}

Now, we move to the ITE estimator in Algorithm \ref{alg:factual-counterlBal}.
\begin{theorem}
\label{thm:alg_2_error_bound_1}
Assume \ref{ass:unconfoundedness}, \ref{ass:sameDistr}, and \ref{ass:weight_exist}. For $w \in \{t, c\}$, the following holds with probability at least $(1 - \delta)$:
\begin{equation*}
	R_{p_w} [\widehat{\beta}_w^{final}] \leq R_{p_w} [\beta_{p_w, \lambda_{p_w}}^*] + \Delta_{S_f}^w + \Delta_{R_f}^w,
\end{equation*}
where 
\begin{equation*}
	\Delta_{S_f}^w = \frac{L^2 \kappa^{\frac{3}{2}} \sigma^{\frac{1}{2}}(\mathbf{K}_w)}{n_w \lambda_w^\prime} \sqrt{\frac{(C n_w)^2 + 1}{n_w}} \left(1 + \sqrt{2 \log \left(\frac{6}{\delta}\right)}\right)
\end{equation*}
\vskip-0.75cm
\begin{multline*}
	\Delta_{R_f}^w = \frac{4 (C n_w) L \nu^\prime_w \sqrt{tr(\mathbf{K}_w)}}{n_t} + \\ 6 (C n_w) L \sqrt{\frac{\log(6 / \delta)}{2n_w}} + C L || \widehat{\beta}_{q_t, \lambda_{q_t}} ||_2
\end{multline*}
with 
\begin{equation*}
	\nu^\prime_w = \nu_w + \sqrt{\frac{4L}{\lambda_{q_w}} \left(\frac{2 \nu_{q_w} \sqrt{tr(\mathbf{K}_w)}}{n_w} + 3 \sqrt{\frac{\log (6 / \delta)}{2 n_w}}\right)}.
\end{equation*}
Here, the notations are defined as in Theorem \ref{thm:alg_1_error_bound_1}, $\nu_w = \big|\big|\beta_{p_t, \lambda_{p_t}}^* - \beta_{q_t, \lambda_{q_t}}^*\big|\big|_2^2$, and $\nu_{q_w} = \big|\big|\beta_{q_w, \lambda_{q_w}}^*\big|\big|_2^2$.
\end{theorem}
\begin{proof}
This is an application of Theorem $6$ in \cite{reddi2015doubly} (full paper) with balancing weights $(n_t \eta)$ and $(n_c \gamma)$, hyperparameters $(n_t \lambda_{p_t})$ and $(n_c \lambda_{p_c})$, and bounds $(C n_t)$ and $(C n_c)$.
\end{proof}
Similarly to Theorem \ref{thm:alg_1_error_bound}, we can bound the error of the ITE estimator in Algorithm \ref{alg:factual-counterlBal}. 
\begin{theorem}
\label{thm:alg_2_error_bound}
Under Assumptions \ref{ass:unconfoundedness}, \ref{ass:sameDistr}, \ref{ass:weight_exist}, \ref{ass:indep_PO}, \ref{ass:true_form_PO}, the following holds with probability at least $(1 - \delta)$:
\begin{multline*}
	\mathbb{E}_p \norm{(Y(1) - Y(0)) - (X \widehat{\beta}_t^{final} - X \widehat{\beta}_c^{final})}_2^2 \leq \\
    2\mathbb{E}_p \norm{(Y(1) - Y(0)) - (X \beta_{p_t, \lambda_{p_t}}^* - X \beta_{p_c, \lambda_{p_c}}^*)}_2^2 \\
    + 2\left(\Delta_{S_f}^t + \Delta_{R_f}^t + \Delta_{S_f}^c + \Delta_{R_f}^c\right),
\end{multline*}
where notations are defined in Theorem \ref{thm:alg_2_error_bound_1}.
\end{theorem}
\begin{proof}
We can apply exactly the same technique as in the proof of Theorem \ref{thm:alg_1_error_bound} to establish this bound. Also, a similar remark could give another error bound. 
\end{proof}

\section{Related Work}
\label{sec:relate_work}

Causal inference has been a major research area across different fields, ranging from economics through statistics, sociology, education, medicine to machine learning (see \cite{rubin2011causal, morgan2014counterfactuals, tian2014simple, chernozhukov2013inference, van2007causal, langford2011doubly, scholkopf2012causal} among others). 

Many methods have been proposed to give average treatment effect estimates in both experimental and observational studies (see \cite{imbens2015causal} for an overview). Another parallel focus is on the estimation of ITEs. Recently, there is a new trend in estimating ITEs with high-dimensional observational data. The methods in this direction include nearest-neighbor matching, Inverse Propensity Score Weighting \cite{rosenbaum2002observational, austin2011introduction}, and notably Ordinary Least Square (OLS) with regularization (such as Elastic-net \cite{hastie2011elements}), Bayesian Additive Regression Trees (BART) \cite{chipman2010bart}, Causal Forest \cite{wager2016estimation}, and Counterfactual regression (CFR) by learning balanced representations \cite{johansson2016learning}.  

OLS with regularization consists of estimating the regularized models for the treated and control data separately and take difference between the two. 

Motivated by boosting algorithms, BART is a nonparametric Bayesian regression using dimensionally adaptive random basis elements. This method has been used successfully in the past for causal inference learning \cite{hill2011bayesian}.

Causal Forest is an extension and modification of the traditional random forest algorithm \cite{breiman2001random}. This method consists of many causal trees, each of which estimates the treatment effect at the leaves \cite{athey2016recursive}.

The CFR by learning balanced representation method is based on using deep learning on an objective function comprising of both factual and proxy counterfactual errors and the discrepancy distance between the distributions of the treated and control groups.

\section{Experiments}
\label{sec:experiment}

In this section, we use simulated data to demonstrate the performance of the proposed method relative to the state-of-the-art approaches. 

\subsection{Simulation Settings}
We consider three simulated datasets.

\subsubsection{Linear Outcome Model}
The first simulation has complex propensity score model but simple linear outcome model; also, both the propensity score and outcome models are sparse. Specifically, the data $(Y_i, W_i, X_i)_{i = 1}^n$ has for each $i$: 
		\begin{itemize} 
			\item $X_i \sim N(0, I_{p \times p})$; 
			\item $W_i \sim Bernoulli (\theta_i)$ where $\theta_i = 1 - 1 / (1 + e^{X_i \beta_W})^{1.23}$ with $\beta_W = (\underbrace{1, ..., 1}_\text{$10$}, \underbrace{0, ..., 0}_\text{$p-10$})$; 
			\item $Y_i = X_i \beta + 10W_i + \epsilon_i$ where $\beta_j = 1 / j^2$ for $j = 1, ..., p$. Here, $\epsilon_i$'s are i.i.d. $N(0, 1)$. 
		\end{itemize}
We use $p = 100$ and $n = 1500$ in which $500$ are for testing. 

\subsubsection{Complex Outcome Model}
The second simulation has both the propensity score and outcome models being complex. In this simulation, the propensity score model is sparse but the outcome model is dense. More precisely, the data $(Y_i, W_i, X_i)_{i = 1}^n$ consists of:
		\begin{itemize}
			\item $X_i \sim N(0, I_{p \times p})$;
			\item $W_i \sim Bernoulli\big(1 - e^{-\theta_i}\big)$ where $\theta_i = 0.89 \times \log \left(1 + e^{-2 - 2X_i \beta_W}\right)$ with $\beta_W = (\underbrace{1, ..., 1}_\text{$10$}, \underbrace{0, ..., 0}_\text{$p-10$})$;
			\item $Y_i = X_i \beta + \theta_i (2W_i - 1) / 2 + \epsilon_i$ where $\beta_j = 1$ for $j = 1, ..., p$. Here, $\epsilon_i$'s are i.i.d. $N(0, 1)$.  
		\end{itemize}
We also use $p = 100$, $n = 1500$ with $500$ testing units. 

\subsubsection{IHDP based Simulation}
We generate a dataset based on \cite{hill2011bayesian} who in turn uses a dataset based on the Infant Health and Development Program (IHDP). The original IHDP dataset is formed from a randomized control experiment. \cite{hill2011bayesian} uses the same covariates from this IHDP dataset, simulates outcomes and introduces an imbalance between the treated and control groups by removing a subset of the treated group. In total, \cite{hill2011bayesian} uses a dataset of 747 individuals in which 139 are treated and there are 25 covariates. More details can be found in that paper. We modify the outcome formula in setting B of \cite{dorie2016nonparametrics} while keeping other things unchanged:
\begin{equation*}
	Y(0) = \sqrt{\exp([\textbf{1}, (X + 0.5)] \times \beta_0)} + N(0, 1)
\end{equation*}
and
\begin{equation*}
    Y(1) = [\textbf{1}, (X + 0.5)] \times \beta_1 + N(0, 1)
\end{equation*}
where $(\beta_0)_j = 0$ with probability 0.6, and equals each of $\{0.1, 0.2, 0.3, 0.4\}$ with probability $0.1$ for $j = 1, ..., 26$; and $(\beta_1)_j = 1 / j$ for $j = 1, ..., 26$.  

Due to non-random removal of the treated population, the propensity score model is not simple. The outcome model is also complex and sparse. The histogram of the ITE of this dataset is provided in Figure \ref{fig:IHDP_based_data}. 

\begin{figure}[t!]
\centering
\includegraphics[width=1.0\columnwidth]{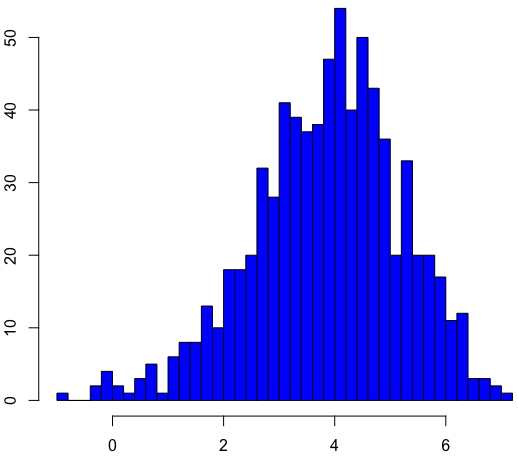}
\caption{Histogram of ITE in IHDP Based Data.}
\label{fig:IHDP_based_data}
\end{figure}

\subsection{Comparison Metrics}
We have a hold-out simulated test set: $(Y^{te}_i$, $Y_i^{te}(1)$, $Y_i^{te}(0)$, $W_i^{te}$, $X_i^{te})_{i = 1}^{n_{te}}$. 

We compare the methods in terms of Precision in Estimation of Heterogeneous Effect (PEHE), where PEHE is
\begin{equation*}
	\sqrt{\frac{1}{n_{te}} \sum_{i = 1}^{n_{te}} \big(\big(\widehat{Y}_i^{te}(1) - \widehat{Y}_i^{te}(0)\big) - \big(Y_i^{te}(1) - Y_i^{te}(0)\big)\big)^2}.
\end{equation*}
Here, we assume that we observe only $(X^{te}_i)_{i = 1}^{n_{te}}$ on the test set. For both metrics, the estimated values $(\widehat{Y}_i^{te}(1), \widehat{Y}_i^{te}(0))$ are given by the estimation method; for example, the estimator in Algorithm \ref{alg:factual-counterlBal} would use $\widehat{Y}_i^{te}(1) = X_i^{te} \widehat{\beta}_t^{final}$ and $\widehat{Y}_i^{te}(0) = X_i^{te} \widehat{\beta}_c^{final}$. 

The PEHE metric is motivated by the fact that in many cases, we are interested in an ITE estimation of an individual with only observed covariates $X$. One such case is when one needs to make a policy decision on whether an individual should apply the treatment. In such cases, neither $W$ nor $Y$ is observed at the time of decision making. Apparently, a good PEHE requires good estimations of both factual and counterfactual outcomes, or at least the difference between the two, not just that of the counterfactual. 

\subsection{Simulation Results}
The methods we use to compare our proposed method with are currently state-of-the-art in estimating ITEs in high-dimensional observational data: OLS with elastic-net regularization \cite{hastie2011elements}, BART \cite{hill2011bayesian}, Causal Forest (CF) \cite{wager2016estimation}, Counterfactual regression (CFR) by learning balanced representations \cite{johansson2016learning}. We use the provided \textsf{R} packages BayesTree for BART, causalForest for CF, and the provided CFR python code (on TensorFlow) for CFR. We do not tune the parameters but use the default values instead.     

Regarding our proposed methods, we consider both Algorithms \ref{alg:counterfactualBal} and \ref{alg:factual-counterlBal} and both $L_\infty$ and $L_2$ norms for balancing weight estimation. 

\begin{table}[t]
\caption{Method Comparison For ITE Estimation on Simulated Data 1 (Linear Outcome Model)}
\label{tab:dat_1}
\vskip 0.15in
\begin{center}
\begin{small}
\begin{sc}
\begin{tabular}{lc}
\hline
\abovespace\belowspace
Method & Simple Data \\
\hline
\abovespace
OLS + R  			& 0.35 $\pm$ 0.06  \\
BART     			& 0.56 $\pm$ 0.06  \\
C.Forest 			& 0.72 $\pm$ 0.08  \\ 
CFR      			& 0.42 $\pm$ 0.13  \\ \hline
\textbf{Algo 1 - $L_\infty$} & \textbf{0.29 $\pm$ 0.13} \\
Algo 1 - $L_2$      & 0.41 $\pm$ 0.20  \\
Algo 2 - $L_\infty$ & 1.01 $\pm$ 0.10  \\
Algo 2 - $L_2$      & 0.89 $\pm$ 0.09  \\
\hline
\end{tabular}
\end{sc}
\end{small}
\end{center}
\vskip -0.15in
\end{table}

\begin{table}[t]
\caption{Method Comparison For ITE Estimation on Simulated Data 2 (Complex Outcome Model)}
\label{tab:dat_2}
\vskip 0.15in
\begin{center}
\begin{small}
\begin{sc}
\begin{tabular}{lc}
\hline
\abovespace\belowspace
Method & Complex Data \\
\hline
\abovespace
OLS + R  			& 2.90 $\pm$ 0.15  \\
BART     			& 7.71 $\pm$ 0.45  \\
C.Forest 			& 5.12 $\pm$ 0.64  \\ 
CFR      			& 3.62 $\pm$ 1.11  \\ \hline
Algo 1 - $L_\infty$ & 8.23 $\pm$ 1.20  \\
Algo 1 - $L_2$      & 6.79 $\pm$ 1.90  \\
\textbf{Algo 2 - $L_\infty$} & \textbf{2.84 $\pm$ 0.16} \\
\textbf{Algo 2 - $L_2$}      & \textbf{2.79 $\pm$ 0.16} \\
\hline
\end{tabular}
\end{sc}
\end{small}
\end{center}
\vskip -0.15in
\end{table}
\begin{table}[t]
\caption{Method Comparison For ITE Estimation on Simulated Data 3 (IHDP Based Data)}
\label{tab:dat_ihdp}
\vskip 0.15in
\begin{center}
\begin{small}
\begin{sc}
\begin{tabular}{lcc}
\hline
\abovespace\belowspace
Method & IHDP based Data \\
\hline
\abovespace
OLS + R & 0.61 $\pm$ 0.13 \\ 
BART & 0.52 $\pm$ 0.06 \\ 
C.Forest & 0.68 $\pm$ 0.16 \\ 
CFR      & 0.84 $\pm$ 0.23 \\ \hline
Algo 1 - $L_\infty$ & 0.63 $\pm$ 0.13 \\ 
Algo 1 - $L_2$ & 1.20 $\pm$ 0.47 \\ 
\textbf{Algo 2 - $L_\infty$} & \textbf{0.44 $\pm$ 0.09} \\ 
Algo 2 - $L_2$ & 0.92 $\pm$ 0.22 \\ 
\hline
\end{tabular}
\end{sc}
\end{small}
\end{center}
\vskip -0.2in
\end{table}
The results are reported in Tables \ref{tab:dat_1}, \ref{tab:dat_2}, \ref{tab:dat_ihdp} for the three simulated datasets. For the first dataset (with linear outcome model), the proposed Algorithm \ref{alg:counterfactualBal} with $L_\infty$ norm in defining the balancing weights works best while the Algorithm \ref{alg:factual-counterlBal} is extremely weak. However for the last two datasets (with complex outcome model),  Algorithm \ref{alg:factual-counterlBal} (also with $L_\infty$ norm) performs strongly. We can intuitively interpret this as that when the outcome model is simple, the simple weighted OLS with regularization can capture all the variations. When this model is complex, we need two steps to get closer to the true ITE values. Also, in general the $L_\infty$ norm in the balancing weight estimation step tends to perform better than the $L_2$ norm.       

In any dataset, the OLS with Regularization is a very strong baseline and hard to beat. BART does not perform as well as it is claimed in \cite{hill2011bayesian}. This is probably because we use the default hyperparameters. The same is true for the CFR method \cite{johansson2016learning}. Tuning hyperparameters for these two methods would likely help improving the performance greatly. Causal Forest \cite{wager2016estimation} is pretty bad for ITE estimation task too. This may be caused by the fact that we use $500$ trees; increasing the number of trees to a much bigger number, say $10,000$ might help. Of course, tuning the hyperparameters for BART and CFR or using a larger number of trees in Causal Forest would be computationally expensive.

\section{Practical Recommendations}
\label{sec:pracRec}

The question we can naturally raise is which of the two proposed algorithms is better as well as which method we should use in each situation.

First, whichever algorithm has a smaller error bound would be more likely to have better performance, though not guaranteed. If 
$$\nu_w = ||\beta_{p_t, \lambda_{p_t}}^* - \beta_{q_t, \lambda_{q_t}}^*||_2^2 << ||\beta_{p_w, \lambda_{p_w}}^*||_2^2 = \nu_{p_w},$$ 
then the error bound of the estimator in Algorithm \ref{alg:factual-counterlBal} is smaller than that of Algorithm \ref{alg:counterfactualBal}. The problem is that we do not have a clue if this condition holds.

Below, we provide several practical recommendations to practitioners who need to perform ITE estimation tasks. 

\begin{itemize}
	\item If we know for sure that the outcome model is simple (e.g., linear), then we should use Algorithm \ref{alg:counterfactualBal}. Otherwise, it is advisable to always use Algorithm \ref{alg:factual-counterlBal}. 
    
    \item We should use $L_\infty$ norm in determining the balancing weights in the first step of either proposed algorithm as it appears more stable than the $L_2$ norm. 
    
    \item There is no single method that outperforms all others in all settings. It is a good practice to use several state-of-the-art methods to estimate the ITEs. If the estimations returned by these methods are significantly different, then the results returned by any single method is not credible enough for the decision makers to take a confident action. As a related work, \citet{athey2017estimating} propose a set of supplementary analyses for determining the credible ATE estimate.

\end{itemize}

\section{Conclusion}
\label{sec:conclusion}

In this paper, we propose a framework for estimating the ITEs in high-dimensional observational studies. Our framework consists of defining balancing weights between the treated and control groups with the whole population, and using these weights for later steps. Beside theoretical error bound results, we supply an empirical analysis based on three simulated datasets. Our proposed methods perform very well, and even outperform the state-of-the-art approaches in the three considered simulations. Moreover, our methods have the merit of simplicity in terms of implementation with mostly no need for hyperparameter tuning. 

The questions remain, still, include determining which method is the best in each situation as no method would always be superior to all others. Supplementary analyses will likely be required for the practitioners to make credible decisions. 

Another research direction includes considering multiple treatments in which case the balancing weights should be more difficult to define. 

\newpage
{\normalsize{
\bibliography{fcbalance}
\bibliographystyle{icml2016}
}}

\end{document}